\documentclass[journal,onecolumn]{IEEEtran}
\usepackage[utf8]{inputenc} 
\usepackage[T1]{fontenc}
\usepackage{url}
\usepackage{ifthen}
\usepackage{cite}
\usepackage[cmex10]{amsmath}
\usepackage{enumitem}

\usepackage[hidelinks]{hyperref}
\usepackage{amsmath, epsfig, cite}
\usepackage{amsthm}
\usepackage{amsfonts}
\usepackage{graphicx}
\usepackage{soul}
\usepackage{latexsym}
\usepackage{amssymb}
\usepackage{color}
\usepackage{url}
\usepackage{colortbl}
\usepackage{comment}
\usepackage[dvipsnames]{xcolor}
\usepackage{caption}
\usepackage{subcaption}

\usepackage{hyperref}
\usepackage{cleveref}

\usepackage{xfrac}
\usepackage{diagbox}
\usepackage{algorithm}
\usepackage{algpseudocode}
\usepackage{dirtytalk}

\DeclareMathAlphabet{\mathbfsl}{OT1}{ppl}{b}{it} 

\newcommand{\R}{\mathbb{R}}

\newcommand{\E}{\mathbb{E}}

\newcommand{\cO}{{\cal O}}

\newcommand{\cS}{{\cal S}}

\newcommand{\bfp}{{\boldsymbol p}}

\newcommand{\bfs}{{\boldsymbol s}}

\newcommand{\bfx}{{\boldsymbol x}}

\theoremstyle{definition}
\newtheorem{theorem}{Theorem}
\newtheorem{lemma}{Lemma}
\newtheorem{remark}{Remark}

\newtheorem{corollary}{Corollary}
\newtheorem{definition}{Definition}

\newtheorem{example}{Example}

\newtheorem{problem}{Problem}
\newtheorem{claim}{Claim}

\newcommand{\abs}[1]{|#1|}

\newcommand{\set}[2]{\left\{#1\;\left|\; #2\right.\right\}}

\title{\textbf{Analyzing Collection Strategies: A Computational Perspective on the Coupon Collector Problem
}\vspace{0ex}}
\author{%
  \IEEEauthorblockN{Hadas~Abraham}
  \IEEEauthorblockA{The Henry \& Marilyn Taub faculty of Computer Science\\
                    Technion\\
                    Haifa, Israel\\
                    Email: hadasabraham@campus.technion.ac.il}
\and
 \IEEEauthorblockN{Ido~Feldman}
  \IEEEauthorblockA{The Henry \& Marilyn Taub faculty of Computer Science\\
                    Technion\\
                    Haifa, Israel\\
                    Email: ido19311@gmail.com}
  \and
 \IEEEauthorblockN{Eitan~Yaakobi}
  \IEEEauthorblockA{The Henry \& Marilyn Taub faculty of Computer Science\\
                    Technion\\
                    Haifa, Israel\\
                    Email: yaakobi@cs.technion.ac.il}

}

\author{\IEEEauthorblockN{\textbf{Hadas Abraham}\IEEEauthorrefmark{2}, \textbf{Ido Feldman}\IEEEauthorrefmark{2}, and \textbf{Eitan Yaakobi}\IEEEauthorrefmark{2}}
\vspace{-.29ex}
\IEEEauthorblockA{\IEEEauthorrefmark{2}{The Henry and Marilyn Faculty of Computer Science, Technion -- Israel Institute of Technology, Haifa, Israel.}}
\vspace{-.29ex} \IEEEauthorblockA { Emails: hadasabraham@campus.technion.ac.il, ido19311@gmail.com, yaakobi@cs.technion.ac.il}
\vspace{-3.29ex}
}

\IEEEoverridecommandlockouts
\IEEEaftertitletext{\vspace{0ex}}

\IEEEoverridecommandlockouts
\IEEEaftertitletext{\vspace{0ex}}

\begin{document} 

\maketitle

\begin{abstract}
The Coupon Collector’s Problem (CCP) is a well-known combinatorial problem that seeks to estimate the number of random draws required to complete a collection of $n$ distinct coupon types. Various generalizations of this problem have been applied in numerous engineering domains. However, practical applications are often hindered by the computational challenges associated with deriving numerical results for moments and distributions. In this work, we present three algorithms for solving the most general form of the CCP, where coupons are collected under any arbitrary drawing probability, with the objective of obtaining $t$ copies of a subset of $k$ coupons from a total of $n$. The First algorithm provides the base model to compute the expectation, variance, and the second moment of the collection process. The second algorithm utilizes the construction of the base model and computes the same values in polynomial time with respect to $n$ under the uniform drawing distribution, and the third algorithm extends to any general drawing distribution. All algorithms leverage Markov models specifically designed to address computational challenges, ensuring exact computation of the expectation and variance of the collection process. Their implementation uses a dynamic programming approach that follows from the Markov models framework, and their time complexity is analyzed accordingly.
\end{abstract}
\section{Introduction}\label{sec:intro}
The classic \emph{Coupon Collector's Problem (CCP)}~\cite{feller1991introduction,flajolet1992birthday} assumes that there are $n$ different types of coupons and the question of interest is \emph{how many coupons one should collect before possessing one coupon of each type}. It is well known that if the coupons are drawn uniformly at random (with repetition), then the expected number of coupons necessary 
to have at least one coupon from each type is roughly $n\log n$. The CCP has far-reaching applications in computational complexity, learning models, probabilistic algorithms, and various other biological sciences and engineering applications~\cite{boneh1996general}. Furthermore, it plays a critical role in applications ranging from system failure analysis to optimal sampling methods in various scientific fields, such as computer science and optimization. The CCP has been extensively studied across these domains due to its broad utility in solving complex engineering and probabilistic problems \cite{boneh1997coupon}.


Several other variants of the CCP have been explored in the literature which considers, for example, the number of collected coupons ($k$), the number of copies for each collected coupon ($t$), and the probability distribution to draw each coupon \cite{anceaume2015new,berenbrink2009weighted,doumas2012coupon,flajolet1992birthday,holst1986birthday,neal2008generalised,von1954coupon}.
This paper is focused on the most \textbf{general version} of the CCP which studies the required number of collected coupons one should draw before collecting $t$ copies of $k$ distinct coupons with any drawing probability. In fact, finding the expected number of drawn coupons for this general setup was solved in~\cite{boneh1997coupon,erdHos1961classical, flajolet1992birthday,klamkin1967extensions,newman1960double} (see this formula in Equation (\ref{eq:gen_t_k})).
However, computationally calculating this formula is not feasible even for a moderate value of $n,k,t$  due to the computational overflows, rounding errors, and time complexity. For example, in~\cite{boneh1989coupon}, it was claimed that \say{It appears that only rarely and then, for rather small $n$ any expectations or probabilities have been explicitly calculated}, and it was also reported in~\cite{boneh1989coupon} that \say{The situation seemed so bad that a researcher in an area that applies the CCP complained that once $n$ exceeds $30$ or so, it is immaterial whether one knows the probabilities or not - since nothing can be computed with them anyway.}

The main contribution of this paper is overcoming these major drawbacks by proposing efficient methods to exactly calculate the expectation and variance of the collection process. In particular, when the probability distribution is uniform, then calculating the exact expectation according to Equation (\ref{eq:uni_t_k}) is exponential with $n$ for any $t,k>1$, while the complexity of our proposed algorithm is $O(n^t)$.

Due to the numerical challenges of exact computation, asymptotic estimations are of great importance. Some asymptotic results for the case of the classic CCP with uniform drawing probability have been obtained by several studies~\cite{baum1965asymptotic,holst1986birthday,janson1983limit, newman1960double}. However, for any drawing probability, asymptotic estimations have been obtained only for a few specific cases when $\bfp$ is characterized by several parameters. Thus, in the general settings of the CCP, to the best of our knowledge, there is no efficient way to explicitly calculate or rather estimate the expectation of the collection process.

The paper is organized as follows:  \autoref{sec:defprobrelated} provides definitions, problem statement, and discusses some relevant prior results on the CCP.
\autoref{sec:nonunidist} presents the base algorithm for the general CCP with time complexity analysis. 
Subsequently, \autoref{sec:unidist} proposes an algorithm for the general CCP under a uniform distribution with time complexity of $O(n^t)$. Next \autoref{sec:groupsym} extends this algorithm to any probability distribution with time complexity analysis.
\autoref{sec:compa}evaluates the algorithms and compares them. Finally, \autoref{sec:conc} summarizes the findings and concludes the paper.

\section{Definitions, Problem Statement, Related Work}\label{sec:defprobrelated}
\subsection{Definitions and Notations}\label{sec:def}
Let $n$ be the number of different coupons, $k$ the desired number of distinct coupons to collect, and $t$ the number of copies needed to retrieve a coupon. 
The probability vector $\bfp \in \R^{n}$ represents the drawing probability with $p_i$ being the probability that coupon $i$ is issued. To analyze the collection process, it will be represented as a discrete-time Markov chain. Each state in this process will be a length-$n$ vector where the $i$-th value represents the number of collected copies of the $i$-th coupon. 
 
Since we require only $t$ copies of each collected coupon, there is no need to track larger values than $t$. This also implies that the Markov chain is finite. This Markov model is formally defined as follows.
\begin{definition}\label{def:bam}
    The Markov model of the coupon collection process consists of the following: 
    \begin{enumerate}
        \item Let $X_b$ be a random variable that represents the state of the collection process \textbf{after} drawing $b$ "tracked" coupons.
        \item Let $\cS_{(n,t)}$ be the set of states such that $\cS_{(n,t)}=\set{\bfs=(s_1,\ldots,s_{n})}{\forall i\in [n]: s_i\leq t}$, where for each $1 \leq i \leq n$, $s_i $ is the number of copies drawn from the $i$-th coupon. For $\bfs\in \cS_{(n,t)}$, let $\widehat{\bfs} \triangleq \sum_{i=1}^{n} s_i$, be the sum of tracked collected coupons at state $\bfs$.
        \item Define \textbf{complete} coupons as those with at least $t$ copies and \textbf{incomplete} coupons as those with less than $t$ copies.
        Let $F(\bfs) \triangleq \{ i \hspace{1mm} | \hspace{1mm} s_i = t \}$ be the set of complete coupons in state $\bfs$.
        Denote level $L(h)$ to be the set of all states with $h$ complete coupons, i.e., $L(h) \triangleq \set{\bfs}{\abs{F(\bfs)}=h}$.
        \item The\textbf{ terminal states} of this Markov chain are all states in $L(k)$ since the process ends at the first time reaching a state with $k$ coupons with at least $t$ copies.
        \item The state $\bfs\in \cS_{(n,t)}$ is called \textit{reachable} from the state $\bfs'$ if $\forall i \in [n]$ it holds that $s_i - s'_i \geq 0$. Furthermore, $A(\bfs)$ is the set of the direct ancestors of the state $\bfs$, that is, $A(\bfs) = \{\bfs' | \bfs  \text{ is reachable from } \bfs' \text{ and } \widehat{\bfs} - \widehat{\bfs'} = 1\}$.
        \item For all $\bfs\in \cS_{(n,t)}$, let $D_{\bfs}$ be the random variable representing the number of draws made from the start of the collection process until reaching state $\bfs$.
        \item Let $M$ denote the transition matrix of this Markov chain where $M_{\bfs', \bfs} = \text{Pr}(X_{b}=\bfs | X_{b-1}=\bfs' )$ for two states $\bfs'$ and $\bfs$.
        Define $M_{\bfs',\bfs}^{(n)} = \text{Pr}(X_{n}=\bfs | X_{0}=\bfs')$, which is the probability of collecting the additional "tracked" coupons from $\bfs'$ to the composition of collected coupons in $\bfs$ (i.e., $\widehat{\bfs}-\widehat{\bfs'}=n$ strands). The $n$-step transition probability matrix is defined as $M^{(n)}\triangleq(M_{\bfs',\bfs}^{(n)})$.
        The transition matrices $M$ and $M^{(n)}$ are analyzed and defined as the \textbf{conditional probabilities} of transitioning to a new state given that an incomplete coupon is drawn.
    \end{enumerate}
\end{definition}
The coupon collection process satisfies the Markovian property, i.e., for a collection of $b$ states say ${\bfs}_0, {\bfs}_1, \ldots, {\bfs}_{b-1}$: $\text{Pr}\left(X_{b}={\bfs}_{b} | X_0={\bfs}_0, \ldots, X_{b-1}={\bfs}_{b-1}\right) = \text{Pr}\left(X_{b}={\bfs}_{b} | X_{b-1}={\bfs}_{b-1}\right)$.
For shorthand, we refer to the initial state as ${\bfs}_0$ (i.e., zero copies for each coupon) and $\text{Pr}(X_0={\bfs}_0) = 1$. At any state $\bfs$, there are two potential scenarios: either remaining in the current state by drawing a complete coupon that already has $t$ copies, or collecting an incomplete coupon and progressing to a new state.

The next claim shows basic properties which follow directly from Definition~\ref{def:bam}.
\begin{claim}    
   Given $\bfs', \bfs \in \cS_{(n,t)}$, it holds that
\begin{enumerate}
    \item  If $\bfs' \notin A(\bfs)$ then $M_{\bfs', \bfs} = 0 $.
    \vspace{1mm}
    \item $M_{\bfs'', \bfs}^{(n)}=\sum_{\bfs' \in A(\bfs)} M_{\bfs'',\bfs'}^{(n-1)} M_{\bfs', \bfs}$ and if $\widehat{\bfs}-\widehat{\bfs''} \neq n$ then $M_{\bfs'', \bfs}^{(n)}=0$.
    \vspace{1.5mm}
    \item For $\bfs',\bfs \in L(k)$ it holds that $M_{\bfs',\bfs}=M_{\bfs,\bfs'}=0$.
    \item For all $\bfs' \in A(\bfs)$ it holds that $I(\bfs',\bfs) = i$ if and only if $s_i - s'_{i} = 1$, which indicates the coupon collected when transitioning from state $\bfs'$ to state $\bfs$.
\end{enumerate} 
\end{claim}

The next definition introduces auxiliary variables that will be used for the analysis of the collection process.
\begin{definition}\label{def:pfs}
    Given the drawing probability vector $\bfp$, for each state $\bfs\in \cS_{(n,t)}$, we define
    \begin{enumerate}
    \item $p^{F}(\bfs) := \sum_{i \in F(\bfs)} p_i$: The probability of drawing a complete coupon.
    \vspace{1mm}
    \item $d(\bfs)$ is the random variable that governs the number of draws required to progress from the state $\bfs$ into a new state. Note that $d(\bfs)$ follows a geometric distribution, that is, $d(\bfs) \sim \text{Geom}(1 - p^{F}(\bfs))$.
\end{enumerate}
\end{definition}
\begin{example}
    Given $n=3,k=2,t=2,\bfp=(0.2,0.3,0.5)$, it 
    \begin{enumerate}
        \item $\cS_{(3,2)}=\set{\bfs=(s_1,s_2,s_3)}{\forall i\in[3]: s_i\leq 2}$
        \item For the state $\bfs=(1,2,0)$, coupon number $2$ is a complete coupon and $1,3$ are incomplete coupons. Moreover $F(\bfs)=\set{i}{s_i=2}=\{2\}$ and $\widehat{\bfs} = 3$, $p^F(\bfs)=0.3$, $\bfs'=(0,2,0)\in A(\bfs)$ and $d(\bfs)\sim \text{Geom}(1 - 0.3)$.
        \item The Terminal states are in $L(2)=\{(2,2,0), (2,2,1), (0,2,2), (1,2,2), (2,0,2), (2,1,2)\}$
    \end{enumerate}
\end{example}

\subsection{Problem Statement}\label{sec:probstate}

In the \emph{general version} of the coupon collection process, $n$ coupons are drawn with replacement according to a probability vector $\bfp$, and the goal is to collect $k$ distinct coupons, each at least $t$ times. Let $\nu^{\bfp}_{t}(n,k)$ denote the required number of draws; for uniform $\bfp$, we write $\nu_t(n,k)$. This paper addresses the following problem:

\begin{problem}
Given $n\geq k\geq 1$, and $\bfp$, compute the exact expectation $\E[ \nu_{t}^{\bfp} (n,k)]$ and variance $Var( \nu_{t}^{\bfp} (n,k))$.
\end{problem}

\subsection{Related Work}
The classical CCP, i.e., $t=1$, was first studied by Feller~\cite{feller1991introduction}, where it was referred to as the \emph{dixie cup problem}. For $n$ equally likely coupons, the expected number of draws to collect all coupons is
$\E[ \nu_1 (n,k=n)] = nH_n =  n \log n + \gamma n + \cO(1)$,  where $H_n$ is the $n$-th harmonic number and $\gamma \sim 0.577$ is the Euler–Mascheroni constant. More generally, according to~\cite{flajolet1992birthday},$\E [\nu_1(n,k)]  = n(H_n-H_{n-k}) $ which, when $n \to \infty$ (in this case, there exists $0<a<1$, such that for $n$ large enough $k < an$), admits the approximation $\E [\nu_1(n,k)] \approx n \log(\frac{n}{n-k}).$

For the general CCP ($t>1$), the problem is closely related to the classical urn model~\cite{erdHos1961classical,newman1960double}. 
With $n$ urns and drawing probabilities $\mathbf{p}=(p_1,\dots,p_n)$, Newman~\cite{newman1960double} showed that when $n\rightarrow \infty$ and under the uniform distribution,
$\E[\nu_t(n,n)] = n\log n + n(t-1)\log\log n + nC_t + o(n),$
where $C_t$ is a constant depending on $t$. 
Erd\H{o}s and R\'enyi~\cite{erdHos1961classical} further proved that $\nu_t(n,n)$ is tightly concentrated around its mean: after 
$n\log n + n(t-1)\log\log n + nx$ draws, the probability that all coupons appear at least $t$ times converges to $e^{-\frac{e^{-x}}{(t-1)!}}$ for $n$ large enough.
Extensions to nonuniform distributions under additional constraints were studied in~\cite{brayton1963asymptotic}, and Flajolet et al.~\cite{flajolet1992birthday} established results for arbitrary discrete distributions, and showed that the expected number of draws required to obtain at least $t$ copies of $k$ out of $n$ coupons is

\begin{equation}
    \label{eq:gen_t_k}
\E[ \nu_{t}^{\bfp} (n,k)]  \hspace{-2.5px}=  \hspace{-2.8px}\sum_{q=0}^{k-1}  \hspace{-1.5px}\int_0^\infty  \hspace{-6.5px}[u^q]  \hspace{-3px}\prod_{i=1}^n\hspace{-0.5px} \Bigl( e_{t-1} (p_i v) +u\cdot\bigl(e^{p_i v} \hspace{-2px}-e_{t-1} (p_i v)  \bigr) \hspace{-0.5px} \Bigr) e^{-v} dv
\end{equation}

\noindent where $e_t(x)=\sum_{i=0}^t \frac{x^i}{i!}$
and for a polynomial $Q(u)$, $[u^q]Q(u)$ is the coefficient of $u^q$ in $Q(u)$. In case $\bfp$ is the uniform distribution, using the result of~\cite{flajolet1992birthday} it follows that

\begin{equation}
    \label{eq:uni_t_k}
\E[ \nu_{t}(n,k)] \hspace{-2.8px}=\hspace{-3px} \sum_{q=0}^{k-1}\hspace{-1.5px}\binom{n}{q}\hspace{-3px} \left[ \int_0^\infty\hspace{-5px} \left( e_{t-1} (\frac{v}{n})\right)^{n-q}\hspace{-3.5px}\left(e^{\frac{v}{n}} \hspace{-2px}-\hspace{-1px}e_{t-1} (\frac{v}{n}) \right)^{q}\hspace{-2.5px}e^{-v} dv\right]\hspace{-1.5px}.
\end{equation}

\noindent As can be seen, for practical purposes, the expressions in (\ref{eq:gen_t_k}) and (\ref{eq:uni_t_k}) and their asymptotic behavior are not easy to calculate or to work with. Moreover, to the best of our knowledge, explicitly calculating the cumulative probability distribution ${P [ \nu_{t}^{\bfp}(n,k) > m]}$ is still open. Furthermore, it was stated in~\cite{boneh1997coupon, boneh1989coupon} that \say{The computation of the integral in equations (\ref{eq:gen_t_k}), (\ref{eq:uni_t_k}) over the sums is enormous and dealing with $2^{n}$ oscillating terms \ldots unless $k$ is a very small number - or similarly close to $n$, there is no essentially simpler way to express this truly complex combinatorial quantity}.

Prior work has established several limits and approximations for the CCP under uniform and nonuniform distributions. Holst~\cite{holst1986birthday} derived the distribution of $\nu_t(n,k)$ and asymptotic bounds on its expectation and cumulative probabilities for the case $k=n$ as $n \to \infty$. Additional limit results for CCP-related random variables under constraints on the total number of draws were presented in~\cite{flatto1982limit}. Asymptotic expressions for the moments, expectation, and variance of $\nu_t^{\bfp}(n,k=n)$ were further developed in~\cite{doumas2013asymptotics, doumas2016coupon, doumas2012coupon}. For nonuniform distributions, Boneh and Hofri~\cite{boneh1997coupon} showed that accurate approximations can be computed in time linear in $t$, while exact computations require exponential time.

Several works analyze CCP variants using Markov chains~\cite{abraham2024covering, anceaume2015new, 10543138}, typically under asymptotic assumptions such as $n \to \infty$ and/or $k = n$. Probabilities are computed via repeated multiplication of the transition matrix, but the exponential growth of the state space makes this approach computationally expensive and prone to numerical instability.

\section{The General CCP Algorithm}\label{sec:nonunidist}
In this section, we present the \emph{Base Algorithm} (\textbf{BA}) for directly computing the expectation and variance of the general CCP. While conceptually simple, we include it for completeness, as, to the best of our knowledge, it does not appear elsewhere in the literature. Monte Carlo simulation is an alternative, but it is neither exact nor scalable.

The algorithm follows the definitions of the base Markov model presented in Section~\ref{sec:def}. The calculation of these values is done for each state $\bfs\in \cS_{(n,t)} $ in the Markov chain, specifically, it computes:
\begin{enumerate}
    \item The transition probabilities $M_{\bfs_0,\bfs}^{(\widehat{\bfs})} = \Pr(X_{\widehat{\bfs}}=\bfs \mid X_0=\bfs_0)$.
    \item The expectation and second moment of the number of draws required to reach $\bfs$, denoted by $\E[D_{\bfs} \mid X_{\widehat{\bfs}}=\bfs]$ and $\E[D_{\bfs}^2 \mid X_{\widehat{\bfs}}=\bfs]$.
\end{enumerate}
The next claim defines the values of the transition matrices.
\begin{claim}
    Given $\bfs\in \cS_{(n,t)}, \bfs'\in A(\bfs)$, it holds that
    \begin{enumerate}
        \item $M_{\bfs',\bfs } = \frac{p_{I(\bfs',\bfs)}}{1 - p^{F}(\bfs')}$.
        \item $M_{\bfs_0,\bfs}^{(\widehat{\bfs})} = \sum_{\bfs' \in A(\bfs)} M_{\bfs_0,\bfs'}^{(\widehat{\bfs'})} \cdot M_{\bfs',\bfs}$.
    \end{enumerate}
\end{claim}

The next lemmas evaluate the values of the expectation and the second moment of the number of draws to reach a state $\bfs$, respectively.
\begin{lemma}\label{cl:BAd}
    Given a state $\bfs\in\cS_{(n,t)}$, it holds that
    \begin{align*}
        \E[D_{\bfs}|X_{\widehat{\bfs}}=\bfs] = \sum_{\bfs' \in A(\bfs)} \frac{M_{\bfs_0,\bfs'}^{(\widehat{\bfs'})} \cdot M_{\bfs',\bfs}}{M_{\bfs_0,\bfs}^{(\widehat{\bfs})}} \cdot\Biggl(\E[D_{\bfs'}|X_{\widehat{\bfs'}}=\bfs'] + \frac{1}{1 - p^{F}(\bfs')}\Biggr).
    \end{align*}

\end{lemma}
\begin{proof}[Proof sketch] 
    The expectation of the number of draws to reach the state $\bfs$, given that $X_{\widehat{\bfs}} = \bfs$, is calculated as the summation over all ancestor states $\bfs' \in A(\bfs)$ of their transition probability to $\bfs$, multiplied by the sum of the expectation to reach $\bfs'$ and the expectation to leave it, which is $\frac{1}{1 - p^F(\bfs')}$.
\end{proof}

\begin{lemma}\label{cl:BAd2}
    Given a state $\bfs\in\cS_{(n,t)}$ it holds that
    \begin{align*}
     \E[D_{\bfs}^2|X_{\widehat{\bfs}}=\bfs] = \hspace{-5px} \sum_{\bfs' \in A(\bfs)} \hspace{-5px}\frac{M_{\bfs_0,\bfs'}^{(\widehat{\bfs'})} \cdot M_{\bfs',\bfs}}{M_{\bfs_0,\bfs}^{(\widehat{\bfs})}}
     \cdot \Bigl(\E[D_{\bfs'}^2|X_{\widehat{\bfs'}}=\bfs']+
     \frac{2}{1 - p^{F}(\bfs')} \cdot \E[D_{\bfs'}|X_{\widehat{\bfs'}}=\bfs']+ \frac{1 + p^{F}(\bfs')}{(1 - p^{F}(\bfs'))^2}\Bigr).
    \end{align*}

\end{lemma}
\begin{proof}
    The second moment of the number of draws to reach the state $\bfs$, given that $X_{\widehat{\bfs}} = \bfs$, is calculated as the summation over all ancestor states $\bfs' \in A(\bfs)$ of their transition probability to $\bfs$, multiplied by the second moment of reaching $\bfs$ through $\bfs'$. Note that the number of draws required to reach state $\bfs'$, denoted as $A$, and the number of draws to transition from $\bfs'$ to $\bfs$, denoted as $B$, are independent random variables and thus it follows that: $\E[(A+B)^2] = \E[A^2] + 2\E[A] \cdot \E[B] + \E[B^2]$. Moreover, since $d(\bfs) \sim \text{Geom}(1 - p^{F}(\bfs))$, the second moment is $\E[d(\bfs)^2] = \frac{1 + p^{F}(\bfs)}{(1 - p^{F}(\bfs))^2}$.
    Applying all together results with the lemma's statement.
\end{proof}
Recall $\nu^{\bfp}_{t}(n,k)$ from \autoref{sec:probstate}. Finally, the next theorem provides the formulas to calculate the expectation and variance of the collection process for the general CCP.

\begin{theorem}\label{th:BA}
    Given $n,k,t,\bfp$ it holds that
    \vspace{1mm}
    \begin{enumerate}
        \item $\E[\nu^{\bfp}_{t}(n,k)] = \sum_{\bfs \in L(k)} M_{\bfs_0,\bfs}^{(\widehat{\bfs})} \cdot \E[D_{\bfs}|X_{\widehat{\bfs}}=\bfs]$.
        \vspace{2mm}
        \item $\E[\nu^{\bfp}_{t}(n,k)^2] = \sum_{\bfs \in L(k)} M_{\bfs_0,\bfs}^{(\widehat{\bfs})} \cdot \E[D_{\bfs}^2|X_{\widehat{\bfs}}=\bfs]$.
        \vspace{3mm}
        \item $Var(\nu^{\bfp}_{t}(n,k)) = \E[\nu^{\bfp}_{t}(n,k)^2] - \E[\nu^{\bfp}_{t}(n,k)]^2$.
        \vspace{2mm}
    \end{enumerate}
\end{theorem}
\begin{algorithm}
\caption{\textbf{BA} Expectation and Second Moment Calculation}\label{alg:calc}
\begin{algorithmic}[1]
\State \textbf{Initialize:}
\State $current\_layer \gets [\bfs_0]$
\State $E \gets 0 \text{ \# the coupon collection process expectation}$
\State $E2 \gets 0$ \text{\# the coupon collection process second moment}
\State $TS \gets []$ \text{\# the set of terminal states}
\State $new\_layer \gets []$

\For{$iter \gets 0$ to $(k \cdot t + (n - k) \cdot (t - 1))$}
    \State $new\_layer \gets 
    \begin{aligned}[t]
        &\text{Create the reachable states}\\
        &\text{from }current\_layer
    \end{aligned}$
    \State $new\_layer \gets 
        \begin{aligned}[t]
            &\text{Update cumulative probability, } \\
            &\text{expectation, and second moment}\\
            &\text{of the new states based on}\\
            &\text{the ancestors in } current\_layer
        \end{aligned}$
    \State $TS \gets \text{Terminal states from } new\_layer$
    \State $E \gets E + \sum\limits_{s \in TS} P(s) \cdot E(s)$
\State $E2 \gets E2 + \sum\limits_{s \in TS} P(s) \cdot E2(s)$

    \State $current\_layer \gets new\_layer \setminus TS$
\EndFor

\State \textbf{Return:} $E$, $E2$
\end{algorithmic}
\label{alg:basealg}
\end{algorithm}
Algorithm~\ref{alg:basealg} is a pseudocode algorithm of the dynamic programming approach that utilizes the variables presented in Theorem~\ref{th:BA}. Note that the maximum number of coupons that can be drawn for any terminal state in level $k$ is exactly $k$ coupons with $t$ copies and the other $n-k$ with $t-1$ copies, that is, $kt+(n-k)(t-1)$.

The time complexity of the \textbf{BA} is derived by the number of edges in the Markov chain. Denote by $V_{BA}(n,k,t)$ the number of vertices corresponding to the Markov chain in the \textbf{BA}. Similarly, let $E_{BA}(n,k,t)$ be the number of edges in the \textbf{BA} Markov chain. 
The next theorem provides the sizes of the vertices and edges set for all $n,k,t$. 
\begin{theorem}\label{th:countsgen}
    Given $n,k,t$ it holds that
    \begin{enumerate}
        \item $V_{BA}(n,k,t) = 
        \sum_{h=0}^{k} \binom{n}{h} \cdot V_{BA}(n-h,n-h,t-1) = \sum_{h=0}^{k} \binom{n}{h} \cdot t^{n-h}$.\vspace{1mm}
        \item $E_{BA}(n,k,t) =n t\cdot V_{BA}(n-1,k-1,t)$.
    \end{enumerate}
\end{theorem}
\begin{proof}
\begin{enumerate}
    \item We let $h$ be the number of complete coupons in every state, so the value of $h$ is between $0$ and $k$. For a given $h$, there are $\binom{n}{h}$ options to select the indices of the $h$ complete coupons. The coupon count of each of the remaining indices must be between $0$ and $t-1$ so it has $t$ options.
    \item Each edge in this Markov chain is determined by first selecting the incomplete drawn coupon along with its respective number of copies, providing \( nt\) possible options. Among the remaining \( n-1 \) coupons there can be up to \( k-1 \) complete coupons, represented by \( V_{BA}(n-1, k-1, t) \). Note that there can not be $k$ complete coupons while updating an incomplete coupon.
\end{enumerate}
\end{proof}
Note that if $k=n$ then we get that $V_{BA}(n,n,t) = (t+1)^n$ and $E_{BA}(n,n,t) = nt (t+1)^{n-1}$. The next corollary provides upper and lower bounds on the number of states and edges in the Markov chain of the \textbf{BA}.
\begin{corollary} Given $n,k,t$ it holds that
\begin{enumerate}
    \item $t^n = V_{BA}(n,n,t-1) \leq V_{BA}(n,k,t) \leq V_{BA}(n,n,t) = (t+1)^n$,
    \item $n(t-1)^{n} < n(t-1) \cdot t^{n-1} = E_{BA}(n,n,t-1)\leq E_{BA}(n,k,t) \leq E_{BA}(n,n,t) =nt \cdot (t+1)^{n-1} < n(t+1)^{n}$.
\end{enumerate}

Therefore, $V_{BA}(n,k,t) = O((t+1)^{n})$ and $E_{BA}(n,k,t) = O((t+1)^{n} \cdot n)$, which is the time complexity of the \textbf{BA}.

\end{corollary}

\section{CCP Algorithm for the Uniform Distribution}\label{sec:unidist}

The main goal of this paper is to be able to compute the expectation of the drawing process efficiently, in the general version of the CCP. The approach in \autoref{sec:nonunidist} provides a dynamic programming algorithm that can be calculated, however, the number of states in this model increases exponentially as $n$ grows. In this section, we introduce the \textbf{UDA} (Uniform Distribution Algorithm), and provide a modified analysis of the Markov model that decreases the number of states when $\bfp$ is the uniform distribution (i.e., $p_i=\frac{1}{n}, 1\leq i \leq n$). Therefore, the computation time decreases significantly and becomes polynomial with $n$. This approach is based on merging symmetric states in the Markov chain.

In the collection process, since all coupons have an equal probability of being drawn, any permutation of a given state is equivalent in terms of both probability and expectation (e.g. the states $(0,1,1)$, $(1,0,1)$ and $(1,1,0)$ are analyzed in the same way). To leverage this symmetry, the states are aggregated into a single representative state, which is referred to as the \emph{multiplicity} vector state. Every state vector can be represented as a multiplicity vector, which indicates the number of coupons with a fixed number of copies. The number of copies is between $0$ and $t$, and we let $x_i$ be the number of coupons with $i$ copies. Specifically, a state $\bfs = (s_1, \dots, s_n)$ is expressed as a multiplicity vector state: $(x_0,x_1,\ldots,x_t)$. This representation satisfies $\sum_{i=0}^{t} x_i = n$. Denote by $mv(\bfs)$ the multiplicity vector representation of $\bfs$ (e.g. given $\bfs = (0,0,0,1,1,2)$ then $mv(\bfs)=(x_0=3,x_1=2,x_2=1,x_3=0)$) and $MV(\cS_{(n,t)})$ as the set of all multiplicity vector states of $\cS_{(n,t)}$.

Recall $d(\bfs)$, since $\bfp$ is the uniform distribution, it holds that for all $\bfs \in L(h), p^{F}(\bfs) = \frac{h}{n}$. Thus, $d(\bfs) \sim \text{Geom}\left(\frac{n-h}{n}\right)$ and $\E\left[d(\bfs)\right] = \frac{n}{n-h}$.
Note that given state $mv(\bfs)=(x_0,x_1,\ldots,x_t) \in L(h)$, it follows that $x_t = h$, since $x_t$ indicates the number of complete coupons (i.e., with $t$ copies). This indicates that the terminal states are $L(k)=\set{\bfx}{x_t=k}$. The next claim defines the values of the transition matrix and the number of reachable states (a single draw). In this case, $I(\bfx',\bfx)=i$ is read with $i$ copies gaining one additional copy, inducing the transition from $\bfx'$ to $\bfx$.
\begin{claim}
Given 
    $\bfs'\in \cS_{(n,t)}, \bfx' = mv(\bfs')=(x_0',x_1',\ldots,x_t') \in MV(\cS_{(n,t)})$
 it holds that
\begin{enumerate}
    \item The size of the set of one-step reachable states from $s'$ is $\abs{\set{\bfx=mv(\bfs)}{\bfs'\in A(\bfs)}}=\sum_{i=0}^{t-1}\mathbb{I}_{ \{ x'_i  > 0 \}}$.
    \item $p^{F}(\bfs') = \frac{x_t'}{n}$.
    \item For all $\bfx\in \set{\bfx=mv(\bfs)}{\bfs'\in A(\bfs)} : 0\leq I(\bfx',\bfx)\leq t-1$ and $ M_{\bfx',\bfx} = \frac{x'_{I(\bfx',\bfx)}}{n - x_t'}$.
    \item For all $\bfx\in \set{\bfx=mv(\bfs)}{\bfs'\in A(\bfs)}: \bfx = (x_0,\ldots,x_{I(\bfx',\bfx)}-1,x_{I(\bfx',\bfx)+1}+1,x_t)$.
\end{enumerate}
\end{claim}

\begin{example}
    Given 
        $\bfs' = (0,0,0,1,1,2)\in \cS_{(6,3)},\bfx'=mv(\bfs')=(x_0'=3,x_1'=2,x_2'=1,x_3'=0)\in MV(\cS_{(6,3)})$
   
    \begin{itemize}
    \item $I(\bfx',\bfx)=0: \bfx = (x_0=2,x_1=3,x_2=1,x_3=0)$.
    \item $I(\bfx',\bfx)=1:\bfx= (x_0=3,x_1=1,x_2=2,x_3=0)$.
    \item $I(\bfx',\bfx)=2: \bfx= (x_0=3,x_1=2,x_2=0,x_3=1)$.
    \end{itemize}
   From state $\bfx'$, there are three possible transitions, each altering one index; e.g., $I(\bfx',\bfx)=0$ increases the coupons with one copy by one and decreases those with zero copies by one.
\end{example}

The rest of the analysis of the values $M_{s_0,s}^{(\widehat{s})}$, $\E[D_{\bfs}|X_{\widehat{s}}=s]$, $\E[D_{\bfs}^2|X_{\widehat{s}}=s]$, remains the same as in \autoref{cl:BAd}, \autoref{cl:BAd2}, and \autoref{th:BA} but with the multiplicity vector states $MV(\cS_{(n,t)})$ as the set of states. The following provides an analysis of the time complexity of the \textbf{UDA}, as discussed in \autoref{sec:nonunidist}. Similarly, let \( V_{UDA}(n,k,t) \) and \( E_{UDA}(n,k,t) \) denote the number of vertices and edges in the \textbf{UDA} Markov chain, respectively. 
We first analyze the case $k=n$, which underlies the general complexity analysis.
\begin{lemma}\label{lem:ver_uni_edge_uni}
Given $n$ and $t$ it holds that
\begin{enumerate}
    \item $V_{UDA}(n,n,t) = \binom{n+t}{t}$.
    \item $E_{UDA}(n,n,t) = t \cdot \binom{n+t-1}{t} = n \cdot \binom{n+t-1}{n}$.
\end{enumerate}
\end{lemma}
\begin{proof}
\begin{enumerate}
    \item From the definition of the multiplicity vector states it holds that the number of states equals the number of solutions over the non-negative integers to the equation $x_0+x_1+\cdots + x_t = n$, and this number is $\binom{n+t}{t}$.
    \item We count the incoming edges according to the value of $x_i$ (in the multiplicity vector state) that is changed on every incoming edge. Assume an incoming edge changes the value of $x_i$, for $1\leq i \leq t$, in the multiplicity vector $(x_0,x_1,\ldots,x_t)$, so this value increases by one. Hence, the number of options is the number of solutions over the non-negative integers to the equation $x_0+x_1+\cdots + x_t = n$, where $x_i\geq 1$ and this number is $\binom{n+t-1}{t}$.
\end{enumerate}
\end{proof}
\begin{theorem}\label{th:countssym}
    Given $n,k,t$ it holds that
    \begin{enumerate}
        \item $V_{UDA}(n,k,t) = V_{UDA}(n,n,t) - V_{UDA}(n-(k+1),n-(k+1),t) = \binom{n+t}{n} - \binom{n+t-(k+1)}{n-(k+1)}$.
        \item 
        $E_{UDA}(n,k,t) = t \cdot \binom{n+t-1}{t} - t \cdot \binom{n+t-k-1}{t}$.
    \end{enumerate}
\end{theorem}
\begin{proof}

\begin{enumerate}
    \item We count the number of states similarly to the proof of \autoref{lem:ver_uni_edge_uni} but under the constraint that $x_t\leq k$. This equals to the value of \( V_{UDA}(n,n,t) \) and then subtracting the number of solutions where $x_t\geq k+1$ which equals to $\binom{n+t-(k+1)}{n-(k+1)}$ resulting in $\binom{n+t}{n} - \binom{n+t-(k+1)}{n-(k+1)}$.
    \item We count the incoming edges similar to the proof of \autoref{lem:ver_uni_edge_uni}, but we exclude all edges originating from states above level $k$ (since level $k$ states are terminal states). Assume an incoming edge changes the value of $x_i$ for $1\leq i \leq t$ by one. Hence, we subtract the number of solutions of the equation $x_0+x_1+\cdots + x_t = n$, where $x_t \geq k$, and this number is $\binom{n+t-k-1}{t}$.  
\end{enumerate}
\end{proof}
The next corollary provides upper and lower bounds on the number of states and edges in \textbf{UDA} Markov chain.
\begin{corollary} Given $n,k,t$ it holds that
\begin{enumerate}
    \item $\binom{n+t-1}{n} \leq V_{UDA}(n,k,t) \leq \binom{n+t}{n}$.
    \item $n \cdot \binom{n+t-2}{n} \leq E_{UDA}(n,k,t) \leq n \cdot \binom{n+t-1}{n}$.
\end{enumerate}

Therefore, $V_{UDA}(n,k,t) = O(n^{t})$ and $E_{UDA}(n,k,t) = O(n^t)$, which is the time complexity of \textbf{UDA}.
\end{corollary}

\section{CCP Algorithm for Any Distribution Utilizing Symmetry}\label{sec:groupsym}
The approach in \autoref{sec:nonunidist} does provide a dynamic programming algorithm that deals with any distribution
$\bfp$ to draw each of the $n$ coupons, but it lacks of exponential time complexity compared to the uniform method in \autoref{sec:unidist} with polynomial time complexity. In this section, we will introduce the \textbf{DPSA} (Drawing Probability Symmetry Algorithm), and show how to utilize the symmetry between coupons with the same drawing probability, to reduce the computation time in the same way as we did in \autoref{sec:unidist}. This is the more probable case where there are several types of coupons that have the same probability. For example, there are two types of coupons, each with a different probability.

\begin{definition}
  To analyze the symmetry in the drawing probability vector $\bfp$, let us denote the following notations:
  \begin{enumerate}
    \item
    $G$ is the number of unique values in the probabilities vector $\bfp$, that is, $G=\abs{\set{p_i\in \bfp}{1 \leq i \leq n}}$.
    \item
    $q_1, \dots, q_G$ are the unique values, that is, $\set{p_i\in \bfp}{1 \leq i \leq n}$, where $\forall 1\leq g \leq G - 1: q_g > q_{g+1} > 0$.
    \item
    For $1 \leq g \leq G : C_g\triangleq \{ i | p_i = q_g , p_i \in \bfp\}$ is the $g$-th set of coupons. Note that $\sum_{g=1}^{G} |C_g| = n$.
  \end{enumerate}
\end{definition}

Each state $\bfs\in \cS_{(n,t)}$ in this Markov chain is described as a vector that contains $G$ \emph{multiplicity} vectors as in \autoref{sec:unidist}. We let $x_{(g,i)}$ be the number of coupons with $i$ draws in the $g$-th set. Specifically, a state $\bfs = (s_1, \dots, s_n)$ is expressed as $G$ multiplicity vectors state: $\Bigl((x_{(1,0)},x_{(1,1)},\ldots,x_{(1,t)}),\ldots,(x_{(G,0)},x_{(G,1)},\ldots,x_{(G,t)})\Bigr)$. This representation satisfies $\forall 1\leq g \leq G : \sum_{i=0}^{t} x_{(g,i)} = |C_g|$. For convenience let us denote $v_g:= (x_{(g,0)},x_{(g,1)},\ldots,x_{(g,t)})$, thus, $\bfs$ can be represented as  $(v_1 \dots, v_G)$. Denote $mv_G(\bfs)$ as the $G$ multiplicity vectors representation of $\bfs$ and $MV_G(\cS_{(n,t)})$ as the set of all $G$ multiplicity vectors states of $\cS_{(n,t)}$.
Note that given a state $\bfs \in L(h)$ it follows that $h = \sum_{g=1}^{G} x_{(g,t)}$. Since for all $1\leq g \leq G, x_{(g,t)}$ indicates the number of complete coupons (i.e., with $t$ copies). This indicates that the terminal states are $L(k)=\set{\bfx}{\sum_{g=1}^{G} x_{(g,t)}=k}$.

Recall $d(\bfs)$, since $\bfp$ is the uniform distribution, it holds that for all $\bfs \in L(h), p^{F}(\bfs) = \frac{h}{n}$. Thus, $d(\bfs) \sim \text{Geom}\left(\frac{n-h}{n}\right)$ and $\E\left[d(\bfs)\right] = \frac{n}{n-h}$.
Note that given state $mv(\bfs)=(x_0,x_1,\ldots,x_t) \in L(h)$, it follows that $x_t = h$, since $x_t$ indicates the number of complete coupons (i.e., with $t$ copies). This indicates that the terminal states are $L(k)=\set{\bfx}{x_t=k}$. The next claim defines the values of the transition matrix and the number of reachable states (a single draw). In this case, $I(\bfx',\bfx)=i$ is read with $i$ copies gains one additional copy, inducing the transition from $\bfx'$ to $\bfx$.

The next claim defines the values of the transition matrix and the number of reachable states (a single draw). In this case, $I(\bfx',\bfx)=(g,i)$ is read with $i$ copies from group $g$ that gain one additional copy, inducing the transition from $\bfx'$ to $\bfx$.
\begin{claim}
Given $\bfs' \in \cS_{(n,t)},\bfx'=mv_G(\bfs')=\Bigl((x'_{(1,0)},x'_{(1,1)},\ldots,x'_{(1,t)}),\ldots,(x'_{(G,0)},x'_{(G,1)},\ldots,x'_{(G,t)})\Bigr)\in MV_G(\cS_{(n,t)})$ it holds that

\begin{enumerate}
    \item $\abs{\set{\bfx=mv_G(\bfs)}{\bfs'\in A(\bfs)}}=\sum_{g=1}^{G} \sum_{i=0}^{t-1}\mathbb{I}_{ \{ x'_{(g,i)}  > 0 \}}$.
    \item $p^{F}(\bfs') = \sum_{g=1}^{G} x'_{(g,t)} \cdot q_g$.
    \item $\forall \bfx\in \set{\bfx=mv_G(\bfs)}{\bfs'\in A(\bfs)} : I(\bfx',\bfx) \in \{( 1\leq g \leq G,0\leq i\leq t-1)\}$ and w.l.o.g $ M_{\bfx',\bfx} = \frac{x'_{(g,i)} \cdot q_g}{1 - p^{F}(\bfs')}$.
    \item \(\forall \bfx\in \set{\bfx=mv_G(\bfs)}{\bfs'\in A(\bfs)} : \bfx = \Bigl(v_1,\ldots, (x'_{(g,0)},\ldots,x'_{I(\bfx',\bfx)}-1,x'_{I(\bfx',\bfx)+1}+1,x'_{(g,t)}),\ldots,v_G\Bigr) \).  
\end{enumerate}
\end{claim}

The rest of the analysis of the values $M_{s_0,s}^{(\widehat{s})}$, $\E[D_{\bfs}|X_{\widehat{s}}=s]$, $\E[D_{\bfs}^2|X_{\widehat{s}}=s]$, remains the same as in \autoref{cl:BAd}, \autoref{cl:BAd2}, and \autoref{th:BA}  but with the $G$ multiplicity vectors states $MV_G(\cS_{(n,t)})$ as the set of states. The following provides an analysis of the time complexity of the \textbf{DPSA}, as discussed in \autoref{sec:nonunidist}. Similarly, let \( V_{DPSA}(n,k,t) \) and \( E_{DPSA}(n,k,t) \) denote the number of vertices and edges in the \textbf{DPSA} Markov chain, respectively. For convenience, we define \( c_g := |C_g| \).

\begin{lemma}\label{num-vetrtices-non-uniform}
Given $n,t\text{ and }\bfp$ it holds that $V_{DPSA}^{\bfp}(n,n,t) = \prod_{g=1}^{G} \binom{c_g+t}{t}$.
\end{lemma}

\begin{proof}
From the definition of the $G$ multiplicity vectors state and as a generalization to \autoref{lem:ver_uni_edge_uni} proof it holds that the number of states is the multiplication of the number of solutions for each set $g$, i.e., $\binom{c_g + t}{t}$.
\end{proof}

\begin{lemma}\label{num-edges-non-uniform}
Given $n,t\text{ and }\bfp$ it holds that $E_{DPSA}^{\bfp}(n,n,t) = t\cdot \biggl(\sum_{g=1}^{G} \frac{c_g}{c_g+t}\biggr) \cdot \biggl(\prod_{g=1}^{G} \binom{c_g+t}{t}\biggr)$.
\end{lemma}
\begin{proof}
We count the incoming edges according to the value of $x_{(g,i)}$ (in the $G$ multiplicity vectors state) that is changed on every incoming edge. Assume an incoming edge changes the value of $x_{(i,g)}$, for $(1\leq g \leq G,1\leq i \leq t)$, in the multiplicity vector $v_g$, so this value increases by one. Hence, the number of options is the multiplication of the number of solutions over the non-negative integers to the equation $(x_{(g',0)}+x_{(g',1)}+\cdots + x_{(g',t)}) = c_{g'}$, for every $g'\neq g$ times the number of solutions to $(x_{(g,0)}+x_{(g,1)}+\cdots + x_{(g,t)}) = c_{g}$ where $x_{(g,i)}\geq 1$. Hence we get $\sum_{g=1}^{G} t \cdot \binom{c_g + t - 1}{t} \cdot \biggl(\prod_{g' = 1, g' \neq g}^{G} \binom{c_g' + t}{t}\biggr)$. Since \( \binom{c_g + t - 1}{t} = \binom{c_g + t}{t} \cdot \frac{c_g}{c_g + t} \) we obtain $t\cdot \biggl(\sum_{g=1}^{G} \frac{c_g}{c_g+t}\biggr) \cdot \biggl(\prod_{g=1}^{G} \binom{c_g+t}{t}\biggr)$ .

\end{proof}

One can verify that in the case of uniform distribution, $G=1$ and $c_1 = n$, the formulas are equal to $V_{UDA}$ and $E_{UDA}$, and in the case when there is no symmetry to exploit, $G=n$ and $c_g = 1$, the formulas are equal to $V_{BA}$ and $E_{BA}$.

The next lemma provides an upper bound for the number of states and edges in \textbf{DPSA} Markov chain.
\begin{lemma}\label{lem:upbound}
     Given $n,t\text{ and }\bfp$ it holds that
    \begin{enumerate}
        \item $V_{DPSA}^{\bfp}(n,n,t) \leq \binom{\frac{n}{G}+t}{t}^{G}$,
        \item $E_{DPSA}^{\bfp}(n,n,t) \leq t \cdot \frac{n}{\frac{n}{G} + t} \cdot \binom{\frac{n}{G}+t}{t}^{G}$,
    \end{enumerate}
    where $G$ is the number of unique values in the probability vector $\bfp$.
\end{lemma}

\begin{proof}
\begin{enumerate}
    \item Given $V_{DPSA}^{\bfp}(n,n,t) = \prod_{g=1}^{G} \binom{c_g+t}{t}$ from \autoref{num-vetrtices-non-uniform} it follows that,
 \begin{align*}
        \log \Biggl( \prod_{g=1}^{G} \binom{c_g+t}{t} \Biggr) &=
        \sum_{g=1}^{G} \log \Biggl(  \binom{c_g+t}{t} \Biggr) =
        \sum_{g=1}^{G} \log \Biggl(  \frac{(c_g + t)!}{t! \cdot c_g!}\Biggr) =\\
        &= \sum_{g=1}^{G} \biggl(\log \Biggl(  \frac{(c_g + t)!}{c_g!}\Biggr) - \log(t!)\biggr) = \Biggl( \sum_{g=1}^{G} \sum_{y=1}^{t} \log (c_g + y) \Biggr) - G \cdot \log(t!)
    \end{align*}
Note that $f(x)=\sum_{y=1}^{t} \log (x + y)$ is a concave function since its second derivative is negative,$ \biggl( \sum_{y=1}^{t} \log (x + y)\biggr )'' = \biggl( \sum_{y=1}^{t} \frac{1}{x+y} \biggr) ' = \sum_{y=1}^{t} \frac{-1}{(x+y)^2} < 0 $. Thus by applying Jensen inequality, it holds that,
$\sum_{g=1}^{G} \sum_{y=1}^{t} \log (c_g + y) \leq G \cdot \sum_{y=1}^{t} \log \Bigl( \frac{n}{G} + y \Bigr)$ since $\sum_{g=1}^{G} c_g = n$.
Therefore, it follows that, 
    \begin{align*}
        \log \Biggl( \prod_{g=1}^{G} \binom{c_g+t}{t} \Biggr) &= \Biggl( \sum_{g=1}^{G} \sum_{y=1}^{t} \log (c_g + y) \Biggr) - G \cdot \log(t!) \leq G \cdot \Biggl( \sum_{y=1}^{t} \log \Bigl( \frac{n}{G} + y \Bigr) - \log(t!) \Biggr) =\\
        &= G \cdot \Biggl( \log \biggl( \frac{(\frac{n}{G} + t)!}{(\frac{n}{G})!} \biggr) - \log(t!) \Biggr) = G \cdot \log \Biggl( \binom{\frac{n}{G} + t}{t} \Biggr) = \log \Biggl( \binom{\frac{n}{G} + t}{t}^{G} \Biggr)
    \end{align*}
Since, $\log(\cdot)$ is monotonous ascending, we get $V_{DPSA}^{\bfp}(n,n,t) = \prod_{g=1}^{G} \binom{c_g+t}{t} \leq \binom{\frac{n}{G} + t}{t}^{G}$.
\item Given $E_{DPSA}^{\bfp}(n,n,t) = t \cdot\biggl(\sum_{g=1}^{G} \frac{c_g}{c_g+t}\biggr) \cdot \biggl(\prod_{g=1}^{G} \binom{c_g+t}{t}\biggr)$ from \autoref{num-edges-non-uniform} by applying Jensen inequality it follows that, $\sum_{g=1}^{G} \frac{c_g}{c_g+t} \leq G \cdot \frac{\frac{n}{G}}{\frac{n}{G} + t} = \frac{n}{\frac{n}{G} + t}$ since for any $x > 0$ it hold that $f(x)=\frac{x}{x+t}$ is a concave function $ \biggl( \frac{x}{x+t} \biggr) '' = \biggl( \frac{t}{(x+t)^2} \biggr) ' = \frac{-2t \cdot (x+t)}{(x+t)^4} < 0$. Therefore,
    $$ E_{DPSA}^{\bfp}(n,n,t) = t \cdot\biggl(\sum_{g=1}^{G} \frac{c_g}{c_g+t}\biggr) \cdot  \biggl(\prod_{g=1}^{G} \binom{c_g+t}{t}\biggr) \leq t \cdot\frac{n}{\frac{n}{G} + t} \cdot  \binom{\frac{n}{G} + t}{t}^{G} $$
\end{enumerate}    
\end{proof}

\begin{remark}
    In the case all the sets are the same size ($\frac{n}{G}$), the upper bounds provided in \autoref{lem:upbound} are tight.
\end{remark}
\section{Comparison and Evaluation}\label{sec:compa}
In this section, we will both compare the time complexity of the \textbf{DPSA} for some drawing probability distributions and evaluate the performance of \textbf{UDA} and \textbf{DPSA}.
Let us define the following distributions:
\begin{enumerate}
    \item
    The "$c$ sets" distribution, where $G=c$ and $c_g = \frac{n}{c}$.
    \item
    The "$c$-size set" distribution, where $G = \frac{n}{c}$ and $c_g = c$.
    \item
    The "almost uniform" distribution, where $G=2$ and $c_1 = n-1, c_2=1$.
    \item 
    The "$0$ symmetry" distribution, where $G=n$, $c_g = 1$.
    \item
    The uniform distribution, where $G=1$, $c_1 = n$.
\end{enumerate}

Table~\ref{table:1} presents a comparison of the different distributions discussed above in terms of time complexity. The vertices and edges are computed using the \textbf{DPSA}, which leverages symmetry in the probability vector. For example, in the case of the almost uniform distribution, the \textbf{DPSA} achieves a complexity of $O(n^t)$, whereas without the \textbf{DPSA} approach, the complexity is exponential. Note that the values $q_1, \dots, q_G$ are not relevant to the computational complexity analysis.

\begin{table}[h!]
    \centering
    \renewcommand{\arraystretch}{1.5} 
    \setlength{\tabcolsep}{12pt} 
    
    \Large 
    \caption{Comparison of the number of vertices and edges of the \textbf{DPSA} Markov chain for different distributions which indicates the algorithm time complexity. }
    \label{table:1}
    \begin{tabular}{| c | c | c |}
        \hline
        
        Distribution & $V_{DPSA}(n,n,t)$ & $E_{DPSA}(n,n,t)$ \\ \hline
        
        $0$ symmetry & $(t+1)^n$ & $t \cdot n \cdot (t+1)^{n-1}$ \\ \hline
        
        $c$-size groups & $\binom{c+t}{t}^{\frac{n}{c}}$ & $t \cdot \frac{n}{c + t} \cdot \binom{c+t}{t}^{\frac{n}{c}}$ \\ \hline
        
        $c$-groups & $\binom{\frac{n}{c}+t}{t}^{c}$ & $t \cdot \frac{n}{\frac{n}{c} + t} \cdot \binom{\frac{n}{c}+t}{t}^{c}$ \\ \hline
        
        almost uniform & $(t+1) \cdot \binom{n-1+t}{t}$ & $t \cdot \left(\frac{1}{t+1} + \frac{n-1}{n-1+t}\right) \cdot \left((t+1) \cdot \binom{n-1+t}{t}\right)$ \\ \hline
        
        uniform & $\binom{n+t}{t}$ & $t \cdot \binom{n+t-1}{t}$ \\ \hline
        
    \end{tabular}
\end{table}

Next, we quantitatively evaluate the proposed \textbf{UDA} and \textbf{DPSA} by comparing their analytical results with Monte Carlo simulations\footnote{Monte Carlo simulations estimate uncertain outcomes by repeatedly sampling random variables from a probabilistic model.}. Figure~\ref{fig:preformence} demonstrates the computational efficiency of \textbf{UDA}, showing that exact results that previously considered computationally infeasible are obtained in under five seconds for $k=n$ and $t=3$. Figure~\ref{fig:preformence2} presents analogous results for \textbf{DPSA}, where exact computations are achieved in under five seconds for $k=n$ and $t=2$, considering two coupon types of equal size $0.5n$ with total group probabilities $q_1 = 0.3 \cdot 0.5n$ and $q_2 = 0.7 \cdot 0.5n$.

For both methods, we compare runtime and accuracy against Monte Carlo simulations with $10^3$, $10^4$, and $10^5$ iterations. In all cases, \textbf{UDA} and \textbf{DPSA} consistently outperform the simulations in runtime. While simulations with $10^3$ iterations achieve comparable execution times, they incur the largest error rates, highlighting the tradeoff between accuracy and efficiency inherent to simulation-based approaches and underscoring the advantages of the proposed analytical methods.
\begin{figure}[ht]
    \centering
    \includegraphics[width=\linewidth]{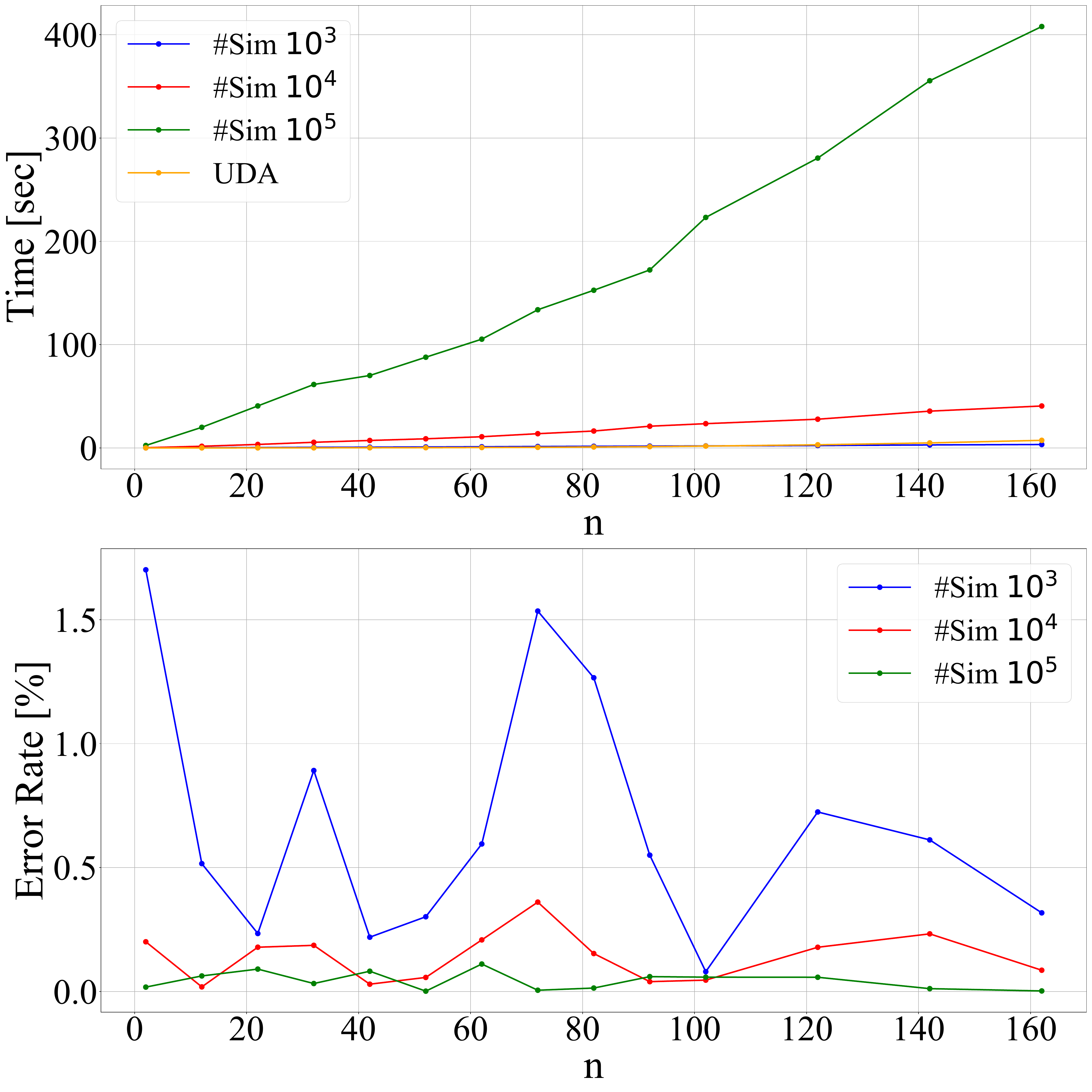}
    \caption{Performance evaluation of Computation Time and Error Rate between the \textbf{UDA} and Monte Carlo simulations. Computation time is shown for increasing values of \( n \), with $k = n$ and for $t=3$.}
    \label{fig:preformence}
\end{figure}
\begin{figure}[ht]
    \centering
    \includegraphics[width=\linewidth]{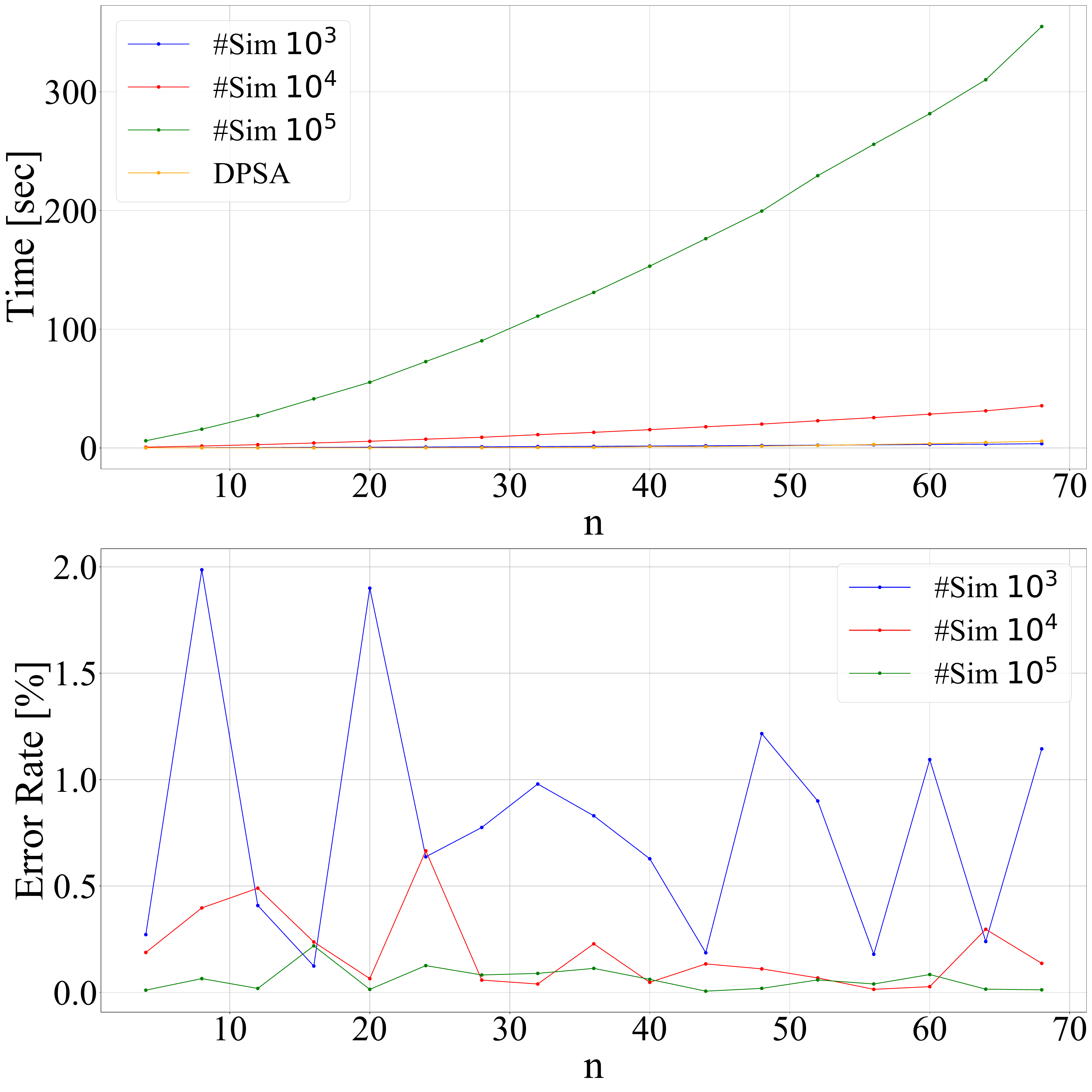}
    \caption{Performance evaluation of Computation Time and Error Rate between the \textbf{DPSA} and Monte Carlo simulations. Computation time is shown for increasing values of \( n \), with $k = n$ , $t=2$ and $g=2$ each group is of size $0.5n$ and $q_1=0.3\cdot 0.5n,q_2=0.7\cdot0.5n$.}
    \label{fig:preformence2}
\end{figure}


\section{Conclusion And Future Work}\label{sec:conc}
In this work, we present a comprehensive modeling of the general coupon collector problem, enabling explicit calculations of the expectation, variance, and second moment of the collection process. We introduce two algorithms, \textbf{UDA} and \textbf{DPSA}, designed to compute these values for both the uniform distribution and any arbitrary distribution, respectively. These algorithms provide polynomial-time exact computations for the uniform distribution. Furthermore, we extend this capability to compute exact values for any probability distribution and generalize the uniform approach to reduce the computational complexity from exponential to polynomial for probability vectors exhibiting symmetry. 
Future work will focus on deriving closed-form expressions for key quantities, reducing reliance on iterative Markov chain computations.

\bibliographystyle{ieeetr}
\bibliography{bib}

@article{erdHos1961classical,
  title={On a classical problem of probability theory},
  author={Erd{\H{o}}s, Paul and R{\'e}nyi, Alfr{\'e}d},
  journal={Magyar Tud. Akad. Mat. Kutat{\'o} Int. K{\"o}zl},
  volume={6},
  number={1},
  pages={215--220},
  year={1961}
}

@book{feller1991introduction,
  title={An introduction to probability theory and its applications, Volume 2},
  author={Feller, William},
  volume={81},
  year={1991},
  publisher={John Wiley \& Sons}
}

@article{flajolet1992birthday,
  title={Birthday paradox, coupon collectors, caching algorithms and self-organizing search},
  author={Flajolet, Philippe and Gardy, Daniele and Thimonier, Lo{\"y}s},
  journal={Discrete Applied Mathematics},
  volume={39},
  number={3},
  pages={207--229},
  year={1992},
  publisher={Elsevier}
}

@article{newman1960double,
  title={The double dixie cup problem},
  author={Newman, Donald J},
  journal={The American Mathematical Monthly},
  volume={67},
  number={1},
  pages={58--61},
  year={1960},
  publisher={JSTOR}
}

@article{boneh1989coupon,
  title={The coupon-collector problem revisited},
  author={Boneh, Arnon and Hofri, Micha},
  year={1989}
}

@article{boneh1997coupon,
  title={The coupon-collector problem revisited—a survey of engineering problems and computational methods},
  author={Boneh, Arnon and Hofri, Micha},
  journal={Stochastic Models},
  volume={13},
  number={1},
  pages={39--66},
  year={1997},
  publisher={Taylor \& Francis}
}

@article{klamkin1967extensions,
  title={Extensions of the birthday surprise},
  author={Klamkin, Murray S and Newman, Donald J},
  journal={Journal of Combinatorial Theory},
  volume={3},
  number={3},
  pages={279--282},
  year={1967},
  publisher={Elsevier}
}

@article{boneh1996general,
  title={General asymptotic estimates for the coupon collector problem},
  author={Boneh, Shahar and Papanicolaou, Vassilis G},
  journal={Journal of Computational and Applied Mathematics},
  volume={67},
  number={2},
  pages={277--289},
  year={1996},
  publisher={Elsevier}
}

@article{baum1965asymptotic,
  title={Asymptotic distributions for the coupon collector's problem},
  author={Baum, Leonard E and Billingsley, Patrick},
  journal={The Annals of Mathematical Statistics},
  volume={36},
  number={6},
  pages={1835--1839},
  year={1965},
  publisher={JSTOR}
}

@article{holst1986birthday,
  title={On birthday, collectors', occupancy and other classical urn problems},
  author={Holst, Lars},
  journal={International Statistical Review/Revue Internationale de Statistique},
  pages={15--27},
  year={1986},
  publisher={JSTOR}
}

@article{janson1983limit,
  title={Limit theorems for some sequential occupancy problems},
  author={Janson, Svante},
  journal={Journal of applied probability},
  volume={20},
  number={3},
  pages={545--553},
  year={1983},
  publisher={Cambridge University Press}
}

@article{brayton1963asymptotic,
  title={On the asymptotic behavior of the number of trials necessary to complete a set with random selection},
  author={Brayton, Robert King},
  journal={Journal of Mathematical Analysis and Applications},
  volume={7},
  number={1},
  pages={31--61},
  year={1963},
  publisher={Elsevier}
}

@article{flatto1982limit,
  title={Limit theorems for some random variables associated with urn models},
  author={Flatto, Leopold},
  journal={The Annals of Probability},
  pages={927--934},
  year={1982},
  publisher={JSTOR}
}

@ARTICLE{10543138,
  author={Preuss, Inbal and Galili, Ben and Yakhini, Zohar and Anavy, Leon},
  journal={IEEE Transactions on Molecular, Biological, and Multi-Scale Communications}, 
  title={Sequencing Coverage Analysis for Combinatorial {DNA}-Based Storage Systems}, 
  year={2024},
  volume={10},
  number={2},
  pages={297-316},
  doi={10.1109/TMBMC.2024.3408053}}

@article{anceaume2015new,
  title={New results on a generalized coupon collector problem using Markov chains},
  author={Anceaume, Emmanuelle and Busnel, Yann and Sericola, Bruno},
  journal={Journal of Applied Probability},
  volume={52},
  number={2},
  pages={405--418},
  year={2015},
  publisher={Cambridge University Press}
}

@article{von1954coupon,
  title={Coupon collecting for unequal probabilities},
  author={Von Schelling, Hermann},
  journal={The American Mathematical Monthly},
  volume={61},
  number={5},
  pages={306--311},
  year={1954},
  publisher={Taylor \& Francis}
}

@article{neal2008generalised,
  title={The generalised coupon collector problem},
  author={Neal, Peter},
  journal={Journal of Applied Probability},
  volume={45},
  number={3},
  pages={621--629},
  year={2008},
  publisher={Cambridge University Press}
}

@inproceedings{berenbrink2009weighted,
  title={The weighted coupon collector’s problem and applications},
  author={Berenbrink, Petra and Sauerwald, Thomas},
  booktitle={International Computing and Combinatorics Conference},
  pages={449--458},
  year={2009},
  organization={Springer}
}

@article{doumas2012coupon,
  title={The coupon collector's problem revisited: asymptotics of the variance},
  author={Doumas, Aristides V and Papanicolaou, Vassilis G},
  journal={Advances in Applied Probability},
  volume={44},
  number={1},
  pages={166--195},
  year={2012},
  publisher={Cambridge University Press}
}

@article{doumas2016coupon,
  title={The coupon collector’s problem revisited: generalizing the double dixie cup problem of Newman and Shepp},
  author={Doumas, Aristides V and Papanicolaou, Vassilis G},
  journal={ESAIM: Probability and Statistics},
  volume={20},
  pages={367--399},
  year={2016}
}

@inproceedings{abraham2024covering,
  title={Covering all bases: The next inning in {DNA} sequencing efficiency},
  author={Abraham, Hadas and Gabrys, Ryan and Yaakobi, Eitan},
  booktitle={2024 IEEE International Symposium on Information Theory (ISIT)},
  pages={464--469},
  year={2024},
  organization={IEEE}
}

@article{doumas2013asymptotics,
  title={Asymptotics of the rising moments for the coupon collector's problem},
  author={Doumas, Aristides and Papanicolaou, Vassilis},
  year={2013}
}
\newpage

\end{document}